\documentclass[british,amssymb, smallextended]{paper}
\usepackage[T1]{fontenc}
\usepackage[latin9]{inputenc}
\usepackage{mathrsfs}
\usepackage{amsthm}
\usepackage{amsmath}
\usepackage{amssymb}
\usepackage{esint}

\makeatletter
\newcommand{\lyxaddress}[1]{
\par {\raggedright #1
\vspace{1.4em}
\noindent\par}
}
\theoremstyle{plain}
\newtheorem{thm}{\protect\theoremname}
  \theoremstyle{plain}
  \newtheorem{cor}{\protect\corollaryname}
  \theoremstyle{plain}
  \newtheorem{lem}{\protect\lemmaname}

\makeatother

\usepackage{babel}
  \providecommand{\lemmaname}{Lemma}
\providecommand{\corollaryname}{Corollary}
\providecommand{\theoremname}{Theorem}

\begin{document}

\title{Divergence of Lubkin's series for a quantum subsystem's mean entropy}

\subtitle{September 2014}

\author{Jacob P Dyer}

\maketitle

\lyxaddress{Department of Mathematics, University of York, York YO10 5DD, UK\\
email: jpd514@york.ac.uk}
\begin{abstract}
In 1978, Lubkin proposed a method of approximating the mean von Neumann
entropy for a subsystem of a finite-dimensional quantum system in
an overall pure state by expanding the entropy as a series in terms
of the mean trace of powers of the system's reduced density operator,
but the convergence of this series was never established. We find
an exact closed form expression for the mean traces, which enables
us to prove that the series converges if and only if the system's
dimension $m\le2$, in spite of the fact that Lubkin's proposed approximation
for the entropy is now known to be correct.\end{abstract}

\section{Introduction}

This paper is a comment on a previous paper by Lubkin \cite{Lubkin1978},
in which he considered the von Neumann entropy of an $m$-dimensional
subsystem $A$ of an $mn$-dimensional quantum system $S$ when $S$
is in a pure state. For a given pure state of $S$ (represented by
the density operator $\hat{\rho}^{mn}$), the entropy is given in
terms of the reduced density operator $\hat{\rho}_{m}^{mn}$ of $A$
as
\[
S_{m,n}=-\text{Tr}[\hat{\rho}_{m}^{mn}\ln\hat{\rho}_{m}^{mn}].
\]
Lubkin was concerned specifically with the mean entropy $\langle S_{m,n}\rangle$
of a random pure state of $S$. Considering the pure state of $S$
instead as a normalised vector $|x\rangle$ in the Hilbert space of
$S$ (denoted $\mathscr{H}_{S}$), such that $\hat{\rho}^{mn}=|x\rangle\langle x|$,
he defined the mean $\langle S_{m,n}\rangle$ with respect to the
natural invariant measure on the unit sphere in $\mathscr{H}_{S}$,
which he referred to as the Haar measure. He attempted to find an
approximation of $\langle S_{m,n}\rangle$ by proposing the use of
the Taylor series expansion of the logarithm, giving (in our notation)
\begin{equation}
\langle S_{m,n}\rangle=\ln m+\sum_{r=1}^{\infty}\frac{m^{r}}{r(r+1)}(-1)^{r}\langle\text{Tr}[(\hat{\rho}_{m}^{mn}-\hat{\rho}_{0})^{r+1}]\rangle,\label{eq:Lubkin-series}
\end{equation}
where $\hat{\rho}_{0}=\hat{\mathbf{1}}/m$. He then showed that
\[
\langle\text{Tr}[(\hat{\rho}_{m}^{mn}-\hat{\rho}_{0})^{2}]\rangle=\frac{m^{2}-1}{mn+1}
\]
and, based on the assumption that (\ref{eq:Lubkin-series}) converged,
truncated the series after the first two terms to propose the approximation
\begin{equation}
\langle S_{m,n}\rangle\simeq\ln m-\frac{1}{2}\frac{m^{2}-1}{mn+1}\label{eq:Lubkin-approx}
\end{equation}
for $n\gg m$ \cite{Lubkin1978,Page1993}.

The purpose of this paper, however, is to show that (\ref{eq:Lubkin-series})
in fact only converges when $m\le2$, and diverges absolutely otherwise.
We do so by finding closed-form expressions for the series terms in
(\ref{eq:Lubkin-series}), given in (\ref{eq:Lubkin-term}), and looking
at their behaviour as $r\rightarrow\infty$. From this it becomes
clear that the series diverges rapidly when $m>2$, indicating that
truncations of it should not be assumed to be good approximations
of the entropy.

This result is unexpected, however, as it is in fact possible to confirm
the validity of Lubkin's approximation for the mean entropy, and even
gain a quantitative measure of the approximation's error, via another
method: in a more recent paper, Page proposed the exact formula
\[
\langle S_{m,n}\rangle=\sum_{k=n+1}^{mn}\frac{1}{k}-\frac{m-1}{2n}
\]
for the entropy \cite{Page1993}, which was later proven by a number
of methods \cite{Foong1994,Sanchez-Ruiz1995,Sen1996}. From this it
is possible to derive the approximation
\[
\langle S_{m,n}\rangle\simeq\ln m-\frac{m^{2}-1}{2mn}
\]
for $n\gg1$, which agrees with -- and improves on -- Lubkin's approximation.
This derivation is discussed in Appendix 3.

\section{Preliminaries\label{sec:Preliminaries}}

Determining where (\ref{eq:Lubkin-series}) does and does not converge
requires studying the large-$r$ limit of $\langle\text{Tr}[(\hat{\rho}_{m}^{mn}-\hat{\rho}_{0})^{r}]\rangle$.
In order to evaluate this, it is easiest to first evaluate $\langle\text{Tr}[(\hat{\rho}_{m}^{mn})^{r}]\rangle$,
which we do in this section using a method based on work by Lloyd
and Pagels \cite{Lloyd1988}, Page \cite{Page1993} and Sen \cite{Sen1996}.
\begin{thm}
\label{thm:real-powers}
\begin{eqnarray}
\langle\text{Tr}[(\hat{\rho}_{m}^{mn})^{r}]\rangle & = & \frac{\Gamma(mn)}{r\Gamma(mn+r)}\sum_{k=0}^{m-1}\frac{(-1)^{k}\Gamma(m+r-k)\Gamma(n+r-k)}{k!\Gamma(r-k)\Gamma(m-k)\Gamma(n-k)}\label{eq:real-powers-expr}
\end{eqnarray}
for any real $r$.\end{thm}
\begin{proof}
To begin with, we assume that $n\ge m$. $\langle\text{Tr}[(\hat{\rho}_{m}^{mn})^{r}]\rangle$
is necessarily symmetric under exchange of $m$ and $n$, as $\hat{\rho}_{m}^{mn}$
and $\hat{\rho}_{\phantom{m}n}^{mn}$ (the reduced density operator
of the `other part' of $S$, which Lubkin refers to as a ``reservoir'')
always have the same eigenvalues when $S$ is in a pure state \cite{Schmidt1907,Ekert1995},
meaning that $\text{Tr}[(\hat{\rho}_{m}^{mn})^{r}]=\text{Tr}[(\hat{\rho}_{\phantom{m}n}^{mn})^{r}]$
for any $r$, so it will still be possible to derive the behaviour
when $n<m$ from these results by exchanging $m$ and $n$ (we will
do so at the end of this proof).

If the eigenvalues of $\hat{\rho}_{m}^{mn}$ are labelled $\{p_{1},\ldots,p_{m}\}$,
then
\[
\text{Tr}[(\hat{\rho}_{m}^{mn})^{r}]=\sum_{i=1}^{m}p_{i}^{r},
\]
which is valid for any real $r$. This is important, as it means that
taking the mean of this expression only requires performing an integral
over the space of possible combinations of eigenvalues (i.e. the space
$(\mathbb{R}^{+})^{m}$). Lloyd and Pagels \cite{Lloyd1988} proved
that (when $n\ge m$) the joint probability distribution over the
eigenvalues which is equivalent to Lubkin's Haar measure is
\[
P(p_{1},\ldots,p_{m})d^{m}p=\Delta^{2}(p_{1},\ldots,p_{m})\delta\left(1-\sum_{i=1}^{m}p_{i}\right)\prod_{k=1}^{m}p_{k}^{n-m}dp_{k},
\]
where
\begin{equation}
\Delta(p_{1},\ldots,p_{m})=\prod_{1\le i<j\le m}(p_{j}-p_{i})=\left|\begin{array}{cccc}
1 & p_{1} & \cdots & p_{1}^{m-1}\\
1 & p_{2} & \cdots & p_{2}^{m-1}\\
\vdots & \vdots & \ddots & \vdots\\
1 & p_{m} & \cdots & p_{m}^{m-1}
\end{array}\right|\label{eq:Vandermonde-1}
\end{equation}
is the Vandermonde determinant of the eigenvalues. Page then showed
how this could be used to construct an eigenvalue integral expression
for $\langle S_{m,n}\rangle$ by integrating $\text{Tr}[\hat{\rho}_{m}^{mn}\ln\hat{\rho}_{m}^{mn}]$
using this distribution \cite{Page1993}.We will now apply the same
method to $\langle\text{Tr}[(\hat{\rho}_{m}^{mn})^{r}]\rangle$. First
we write
\begin{eqnarray*}
\langle\text{Tr}[(\hat{\rho}_{m}^{mn})^{r}]\rangle=\frac{1}{\Lambda}\int\delta\left(1-\sum_{i=1}^{m}p_{i}\right)\Delta^{2}(p_{1},\ldots,p_{m})\\
\times\prod_{k=1}^{m}p_{k}^{n-m}dp_{k}\sum_{i=1}^{m}p_{i}^{r},
\end{eqnarray*}
where $\Lambda$ is a normalisation factor which is a function of
$m$ and $n$ defined such that $\langle\text{Tr}[(\hat{\rho}_{m}^{mn})^{0}]\rangle=m$
\footnote{\unexpanded{\label{fn:a_0}Technically $\langle\text{Tr}[(\hat{\rho}_{m}^{mn})^{0}]\rangle=a_{0}$
where $a_{0}=\min(m,n)$, in keeping with the required symmetry of
$\langle\text{Tr}[(\hat{\rho}_{m}^{mn})^{r}]\rangle$. However, as
we are only looking at the cases where $m\le n$ to begin with, it
is sufficient to say for the moment that $\langle\text{Tr}[(\hat{\rho}_{m}^{mn})^{0}]\rangle=m$.}%
}. Next we multiply this by the factor
\[
\frac{1}{\Gamma(mn+r)}\int_{0}^{\infty}\lambda^{mn+r-1}e^{-\lambda}d\lambda,
\]
(which equals unity by the definition of the gamma function) and perform
the coordinate substitution $q_{i}=\lambda p_{i}$. Some rearrangement
gives
\begin{eqnarray*}
\langle\text{Tr}[(\hat{\rho}_{m}^{mn})^{r}]\rangle=\frac{1}{\Lambda\Gamma(mn+r)}\int_{0}^{\infty}e^{-\lambda}d\lambda\int\delta\left(\lambda-\sum_{i=1}^{m}q_{i}\right)\Delta^{2}(q_{1},\ldots,q_{k})\\
\times\prod_{k=1}^{m}q_{k}^{n-m}dq_{k}\sum_{i=1}^{m}q_{i}^{r}\\
\hphantom{\langle\text{Tr}[(\hat{\rho}_{m}^{mn})^{r}]\rangle}=\frac{1}{\Lambda\Gamma(mn+r)}\int\Delta^{2}(q_{1},\ldots,q_{k})\prod_{k=1}^{m}e^{-q_{k}}q_{k}^{n-m}dq_{k}\sum_{i=1}^{m}q_{i}^{r}.
\end{eqnarray*}
Next, as the integral is symmetric under exchange of any two $q_{i}$,
we can remove the summation and simply write that
\begin{eqnarray}
\langle\text{Tr}[(\hat{\rho}_{m}^{mn})^{r}]\rangle&=\frac{m}{\Lambda\Gamma(mn+r)}\int\Delta^{2}(q_{1},\ldots,q_{m})\prod_{k=1}^{m}\left(q_{k}^{n-m}e^{-q_{k}}dq_{k}\right)q_{1}^{r}\nonumber\\
&=\frac{m}{\Lambda\Gamma(mn+r)}\int\Delta^{2}(q_{1},\ldots,q_{m})\prod_{k=1}^{m}q_{k}^{n-m+r\delta_{k1}}e^{-q_{k}}dq_{k}.\label{eq:Page-integral}
\end{eqnarray}

The remainder of the proof follows the same procedure used by Sen
to prove Page's conjectured entropy formula \cite{Sen1996}. He observed
that, as the determinant of a matrix is unchanged by addition of multiples
of its columns onto each other, the definition of the Vandermonde
determinant given in (\ref{eq:Vandermonde-1}) can be rewritten as
\begin{equation}
\Delta(q_{1,}\ldots,q_{m})=\left|\begin{array}{cccc}
L_{0}^{\alpha}(q_{1}) & L_{1}^{\alpha}(q_{1}) & \cdots & L_{m-1}^{\alpha}(q_{1})\\
L_{0}^{\alpha}(q_{2}) & L_{1}^{a}(q_{2}) & \cdots & L_{m-1}^{\alpha}(q_{2})\\
\vdots & \vdots & \ddots & \vdots\\
L_{0}^{\alpha}(q_{m}) & L_{1}^{\alpha}(q_{m}) & \cdots & L_{m-1}^{\alpha}(q_{m})
\end{array}\right|,\label{eq:Legendre-determinant}
\end{equation}
where $L_{i}^{\alpha}$ are generalised Laguerre polynomials
\[
L_{i}^{\alpha}(q)=\frac{e^{q}}{q^{\alpha}}(-1)^{i}\frac{d^{i}}{dq^{i}}(e^{-q}q^{i+\alpha})
\]
for some real factor $\alpha$, defined such that $L_{i}^{\alpha}(q)$
is an order $i$ polynomial in $q$ where the coefficient of $q^{i}$
is unity (this and any other properties of $L_{i}^{\alpha}$ used
in this paper are taken from \cite{Sen1996}).

Using the expansion of the determinant in terms of the Levi-Civita
symbol, we can write 
\[
\Delta^{2}(q_{1},\ldots,q_{m})=\varepsilon_{i_{1}i_{2}\ldots i_{m}}\varepsilon_{j_{1}j_{2}\ldots j_{m}}\prod_{\sigma=1}^{m}L_{i_{\sigma}}^{\alpha}(q_{\sigma})L_{j_{\sigma}}^{\alpha}(q_{\sigma}).
\]
 If this is substituted into (\ref{eq:Page-integral}) with $\alpha=n-m$,
then the orthogonality relation
\[
\int_{0}^{\infty}q^{n-m}e^{-q}L_{i}^{n-m}(q)L_{j}^{n-m}(q)dq=i!\Gamma(n-m+j+1)\delta_{ij}
\]
causes any terms to vanish which don't satisfy $i_{\sigma}=j_{\sigma}$
for all $2\le\sigma\le m$ (it also follows then that $i_{1}=j_{1}$
by a process of elimination). Collecting together only the non-zero
terms, we are left with
\begin{eqnarray*}
\langle\text{Tr}[(\hat{\rho}_{m}^{mn})^{r}]\rangle=\frac{m!}{\Lambda\Gamma(mn+r)}\sum_{i=0}^{m-1}\int_{0}^{\infty}q_{1}^{n-m+r}e^{-q_{1}}[L_{i}^{n-m}(q_{1})]^{2}dq_{1}\\
\times\prod_{j\ne i}\int_{0}^{\infty}q_{j}^{n-m}e^{-q_{j}}[L_{j}^{n-m}(q_{j})]^{2}dq\\
\hphantom{\langle\text{Tr}[(\hat{\rho}_{m}^{mn})^{r}]\rangle}=\frac{m!}{\Lambda\Gamma(mn+r)}\prod_{j=0}^{m-1}j!\Gamma(n-m+j+1)\\
\times\sum_{i=0}^{m-1}\frac{\int_{0}^{\infty}q^{n-m+r}e^{-q}[L_{i}^{n-m}(q)]^{2}dq}{i!\Gamma(n-m+i+1)}.
\end{eqnarray*}
We fix $\Lambda$ now by looking at the special case $r=0$, where
\[
\langle\text{Tr}[(\hat{\rho}_{m}^{mn})^{0}\rangle=m\cdot\frac{m!}{\Lambda\Gamma(mn)}\prod_{j=0}^{m-1}j!\Gamma(n-m+j+1)=m,
\]
meaning that we can write write
\begin{equation}
\langle\text{Tr}[(\hat{\rho}_{m}^{mn})^{r}]\rangle=\frac{\Gamma(mn)}{\Gamma(mn+r)}\sum_{i=0}^{m-1}\frac{\int_{0}^{\infty}q^{n-m+r}e^{-q}L_{i}^{n-m}(q)L_{i}^{n-m}(q)dq}{i!\Gamma(n-m+i+1)}.\label{eq:sen-integral}
\end{equation}
The remaining integrals do not match the orthogonality relation due
to the additional $q^{r}$ factor, but we can still evaluate them
as finite sums using two additional identities given in \cite{Sen1996}:
\begin{eqnarray}
L_{i}^{\alpha}(q) & = & \sum_{k=0}^{i}\binom{i}{k}(-1)^{k}\frac{\Gamma(i+\alpha+1)}{\Gamma(i+\alpha-k+1)}q^{i-k}\nonumber \\
 & = & \sum_{k=0}^{i}\binom{i}{k}(-1)^{i-k}\frac{\Gamma(i+\alpha+1)}{\Gamma(k+\alpha+1)}q^{k}\label{eq:series-ident}
\end{eqnarray}
and
\begin{eqnarray}
\int_{0}^{\infty}q^{\beta-1}e^{-q}L_{i}^{\alpha}(q)dq & = & (-1)^{i}(1-\beta+\alpha)_{i}\Gamma(\beta)\nonumber \\
 & = & \frac{\Gamma(\beta-\alpha)}{\Gamma(\beta-\alpha-i)}\Gamma(\beta),\label{eq:integral-ident}
\end{eqnarray}
where $(1-a)_{i}=(1-a)(2-a)\ldots(i-a)=(-1)^{i}\Gamma(a)/\Gamma(a-i)$
is the Pochhammer symbol representing the rising factorial. We simplify
(\ref{eq:sen-integral}) by substituting (\ref{eq:series-ident})
in place of one of the $L_{i}^{n-m}$ terms and then evaluating the
integral using (\ref{eq:integral-ident}). This gives
\begin{eqnarray}
\langle\text{Tr}[(\hat{\rho}_{m}^{mn})^{r}]\rangle&=\frac{\Gamma(mn)}{\Gamma(mn+r)}\sum_{i=0}^{m-1}\sum_{k=0}^{i}\frac{(-1)^{i-k}\Gamma(n-m+r+k+1)\Gamma(r+k+1)}{k!\Gamma(i-k+1)\Gamma(n-m+k+1)\Gamma(r+k-i+1)}\nonumber\\
&=\frac{\Gamma(mn)}{\Gamma(mn+r)}\sum_{k=0}^{m-1}\sum_{i=0}^{m-k-1}\frac{(-1)^{i}\Gamma(n-m+r+k+1)\Gamma(r+k+1)}{k!\Gamma(i+1)\Gamma(n-m+k+1)\Gamma(r-i+1)}\nonumber\\
&=\frac{\Gamma(mn)}{\Gamma(mn+r)}\sum_{k=0}^{m-1}\frac{(-1)^{m-k-1}\Gamma(n-m+r+k+1)\Gamma(r+k+1)}{rk!\Gamma(m-k)\Gamma(n-m+k+1)\Gamma(r-m+k+1)}\nonumber\\
&=\frac{\Gamma(mn)}{r\Gamma(mn+r)}\sum_{k=0}^{m-1}\frac{(-1)^{k}\Gamma(m+r-k)\Gamma(n+r-k)}{k!\Gamma(r-k)\Gamma(m-k)\Gamma(n-k)}.\label{eq:m<=n-result}
\end{eqnarray}

The steps in this rearrangement are:\end{proof}
\begin{enumerate}
\item Swap the order of the two summations using $\sum_{i=0}^{m-1}\sum_{k=0}^{i}\equiv\sum_{k=0}^{m-1}\sum_{i=k}^{m-1}$,
then replace $i$ with $i+k$.

\begin{enumerate}
\item Evaluate the sum over $i$ using the identity
\[
\sum_{i=0}^{a}\frac{(-1)^{i}}{\Gamma(i+1)\Gamma(r-i+1)}=\frac{(-1)^{a}}{r\Gamma(a+1)\Gamma(r-a)}
\]
(see Lemma \ref{lem:Gamma-lemma} in Appendix 2).
\item Replace $k$ with $m-k-1$.
\end{enumerate}
\end{enumerate}
\begin{proof}
This result is identical in form to (\ref{eq:real-powers-expr}),
but we have only assumed it to be valid for $n\ge m$ so far. It is
easy to see that it is also valid for $n<m$ though; as stated earlier,
$\langle\text{Tr}[(\hat{\rho}_{m}^{mn})^{r}]\rangle$ is necessarily
symmetric under exchange of $m$ and $n$. Applying this to (\ref{eq:m<=n-result})
gives

\[
\langle\text{Tr}[(\hat{\rho}_{m}^{mn})^{r}]\rangle=\frac{\Gamma(mn)}{r\Gamma(mn+r)}\sum_{k=0}^{n-1}\frac{(-1)^{k}\Gamma(m+r-k)\Gamma(n+r-k)}{k!\Gamma(r-k)\Gamma(m-k)\Gamma(n-k)}
\]
when $n<m$. But because $1/\Gamma(n-k)$ is an entire function with
respect to $k$ with zeroes at all integers $k\ge n$, the upper bound
of the $k$-summation can be raised without changing the result %
\footnote{\unexpanded{This is guaranteed to be true for non-integer $r$. Some
ambiguity can arise for integer $r$ due to divergent terms in the
numerator, but as the integer cases can be treated as the limits of
sequences of non-integer $r$, and the limit of $\langle\text{Tr}[(\hat{\rho}_{m}^{mn})^{r}]\rangle$
is well defined in any such case, it is true for integer $r$ as well.}%
}. Thus, the limit $n-1$ can be increased back to $m-1$ to give the
completely general equation

\[
\langle\text{Tr}[(\hat{\rho}_{m}^{mn})^{r}]\rangle=\frac{\Gamma(mn)}{r\Gamma(mn+r)}\sum_{k=0}^{m-1}\frac{(-1)^{k}\Gamma(m+r-k)\Gamma(n+r-k)}{k!\Gamma(r-k)\Gamma(m-k)\Gamma(n-k)}.
\]

\end{proof}
This derivation followed essentially the same procedure as that used
by Sen to prove Page's exact entropy result \cite{Sen1996}. Sen based
his proof on evaluating $\langle\text{Tr}[\hat{\rho}_{m}^{mn}\ln\hat{\rho}_{m}^{mn}]\rangle$
directly using the same basic method, and the method we used here
can also be used to get the same result using the fact that 
\[
\hat{\rho}_{m}^{mn}\ln\hat{\rho}_{m}^{mn}=\lim_{r\rightarrow1}\frac{\partial}{\partial r}(\hat{\rho}_{m}^{mn})^{r},
\]
which is applicable as (\ref{eq:real-powers-expr}) is valid for any
real value $r$. See Appendix 1 for a demonstration of this.

Next, a further rearrangement of (\ref{eq:real-powers-expr}) is required
to put it into a form suitable for use in Lubkin's series.
\begin{cor}
\begin{eqnarray}
\langle\text{Tr}[(\hat{\rho}_{m}^{mn})^{r}]\rangle & = & \frac{\Gamma(mn)r!}{\Gamma(mn+r)}\sum_{k=0}^{m-1}\binom{r-1}{k}\binom{m}{k+1}\binom{n+r-k-1}{n-1}\label{eq:integer-powers-expr}
\end{eqnarray}
for integer $r\ge1$.\end{cor}
\begin{proof}
First, as only integer $r$ is needed from this point onwards, (\ref{eq:real-powers-expr})
can be restated as
\begin{equation}
\langle\text{Tr}[(\hat{\rho}_{m}^{mn})^{r}]\rangle=\frac{\Gamma(mn)}{r!\Gamma(mn+r)}\sum_{k=0}^{m-1}\binom{r-1}{k}\frac{(-1)^{k}\Gamma(m+r-k)\Gamma(n+r-k)}{\Gamma(m-k)\Gamma(n-k)}.\label{eq:integer-from-real}
\end{equation}

This derivation then relies on the fact that
\[
\frac{\Gamma(m+r-k)}{\Gamma(m-k)}=\left.\frac{\partial^{r}}{\partial u^{r}}\frac{(-1)^{r}}{u^{m-k}}\right|_{u=1},
\]
which follows from simple repeated differentiation. We substitute
this twice into (\ref{eq:integer-powers-expr}) to give
\begin{eqnarray*}
\langle\text{Tr}[(\hat{\rho}_{m}^{mn})^{r}]\rangle & = & \frac{\Gamma(mn)}{r!\Gamma(mn+r)}\left.\frac{\partial^{r}}{\partial u^{r}}\frac{\partial^{r}}{\partial v^{r}}\sum_{k=0}^{m-1}\binom{r-1}{k}\frac{(-1)^{k}}{u^{m-k}v^{n-k}}\right|_{u,v=1}\\
 & = & \frac{\Gamma(mn)}{r!\Gamma(mn+r)}\left.\frac{\partial^{r}}{\partial u^{r}}\frac{\partial^{r}}{\partial v^{r}}\frac{(1-uv)^{r-1}}{u^{m}v^{n}}\right|_{u,v=1},
\end{eqnarray*}
and expand out the derivatives to give
\begin{eqnarray*}
\langle\text{Tr}[(\hat{\rho}_{m}^{mn})^{r}]\rangle=\frac{\Gamma(mn)(-1)^{r}}{r!\Gamma(mn+r)}\sum_{k=0}^{r-1}\binom{r}{k}\frac{(r-1)!}{(r-k-1)!}\\
\times\frac{(n+r-k-1)!}{(n-1)!}\left.\frac{\partial^{r}}{\partial u^{r}}\frac{(1-uv)^{r-k-1}}{u^{m-k}v^{n+r-k}}\right|_{u,v=1}\\
\hphantom{\langle\text{Tr}[(\hat{\rho}_{m}^{mn})^{r}]\rangle}=\frac{\Gamma(mn)}{r!\Gamma(mn+r)}\sum_{k=0}^{r-1}\binom{r}{k}\frac{(n+r-k-1)!}{(n-1)!}\sum_{l=0}^{r-k-1}\binom{r}{l}\\
\times\frac{(r-1)!(m+r-k-l-1)!}{(r-l-k-1)!(m-k-1)!}\left.\frac{(1-uv)^{r-k-l-1}}{u^{m+r-k-l}v^{n+r-k-l}}\right|_{u,v=1}.
\end{eqnarray*}
Taking the limit of $u,v\rightarrow1$ removes all terms except those
where $r-k-l-1=0$, and by collecting together the various factorial
terms we get
\begin{eqnarray*}
\langle\text{Tr}[(\hat{\rho}_{m}^{mn})^{r}]\rangle & = & \frac{\Gamma(mn)r!}{\Gamma(mn+r)}\sum_{k=0}^{r-1}\binom{m}{k+1}\binom{r-1}{k}\binom{n+r-k-1}{n-1}\\
 & = & \frac{\Gamma(mn)r!}{\Gamma(mn+r)}\sum_{k=0}^{m-1}\binom{m}{k+1}\binom{r-1}{k}\binom{n+r-k-1}{n-1}.
\end{eqnarray*}
In the final step here we use the same method for changing limits
that was used at the end of Theorem \ref{thm:real-powers}, which
relies on the fact that $\binom{a}{b}=0$ when $b>a$, meaning that
\[
\binom{m}{k+1}\binom{r-1}{k}=0
\]
when either $k>m$ or $k>r$.
\end{proof}
(\ref{eq:integer-powers-expr}) has the notable property that all
terms in the summation are positive, whereas (\ref{eq:real-powers-expr})
was an alternating summation. However, the main property which motivated
us to use this form is the prefactor, specifically the $r!$ term,
the importance of which will become evident in Theorem \ref{thm:m>=00003D2}.

\section{\label{sec:Special-Case}Example case: $m=n=2$}

Before looking at the series for general dimensions in the next section,
it will be beneficial to first look at the case $m=n=2$. In this
case the terms in (\ref{eq:Lubkin-series}) can be evaluated explicitly
in a simple closed form. To begin with, if we substitute $m=n=2$
into (\ref{eq:integer-powers-expr}) we get
\[
\langle\text{Tr}[(\hat{\rho}_{2}^{2,2})^{r}]\rangle=\frac{6r!}{(r+3)!}(r^{2}+r+2).
\]
The general binomial expansion of$\langle\text{Tr}[(\hat{\rho}_{m}^{mn}-\hat{\rho}_{0})^{r}]\rangle$
is
\[
\langle\text{Tr}[(\hat{\rho}_{m}^{mn}-\hat{\rho}_{0})^{r}]\rangle=\sum_{k=0}^{r}\binom{r}{k}\frac{(-1)^{r-k}}{m^{r-k}}\langle\text{Tr}[(\hat{\rho}_{m}^{mn})^{k}]\rangle,
\]
which in this case means that
\[
\langle\text{Tr}[(\hat{\rho}_{2}^{2,2}-\hat{\rho}_{0})^{r}]\rangle=\sum_{k=0}^{r}\binom{r}{k}(-1)^{r-k}\frac{1}{2^{r-k}}\frac{6k!}{(k+3)!}(k^{2}+k+2).
\]

The next few steps in particular demonstrate the procedure which will
be used on the general case in the next section. First, we rearrange
the above using the identity

\[
\binom{r}{k}\frac{k!}{(k+N)!}=\frac{r!}{(r+N)!}\binom{r+N}{k+N},
\]
giving
\begin{eqnarray*}
\langle\text{Tr}[(\hat{\rho}_{2}^{2,2}-\hat{\rho}_{0})^{r}]\rangle & = & \frac{6r!}{(r+3)!}\sum_{k=0}^{r}\binom{r+3}{k+3}(-1)^{r-k}\frac{1}{2^{r-k}}(k^{2}+k+2)\\
 & = & -\frac{6r!}{(r+3)!}\sum_{k=3}^{r+3}\binom{r+3}{k}\frac{(-1)^{r-k}[k(k-1)-4k+8]}{2^{r-k+3}}.
\end{eqnarray*}
Then, we replace $[k(k-1)-4k+8]$ with
\[
\left.\left(\frac{\partial^{2}}{\partial u^{2}}-4\frac{\partial}{\partial u}+8\right)u^{k}\right|_{u=1},
\]
and finally the summation is replaced with the difference of two sums,
one from zero to $r-3$ (which is a complete binomial expansion),
and the other from zero to two. This gives
\begin{eqnarray*}
\langle\text{Tr}[(\hat{\rho}_{2}^{2,2}-\hat{\rho}_{0})^{r}]\rangle=-\frac{6r!}{(r+3)!}\left(\frac{\partial^{2}}{\partial u^{2}}-4\frac{\partial}{\partial u}+8\right)\left(\frac{(-1)^{r}}{2^{r+3}}\left(1-2u\right)^{r+3}\right.\\
\left.\left.-\sum_{k=0}^{2}\binom{r+3}{k}(-1)^{r-k}\frac{u^{k}}{2^{r-k+3}}\right)\right|_{u=1}\\
\hphantom{\langle\text{Tr}[(\hat{\rho}_{2}^{2,2}-\hat{\rho}_{0})^{r}]\rangle}=\frac{3[1+(-1)^{r}]}{2^{r}(r+3)}.
\end{eqnarray*}

Substituting this into (\ref{eq:Lubkin-series}) then gives
\begin{eqnarray*}
\langle S_{2,2}\rangle & = & \ln2-\sum_{r=1}^{\infty}\frac{3[1-(-1)^{r}]}{2r(r+1)(r+4)}.
\end{eqnarray*}
At this point it is clear that the series will converge for $m=n=2$,
by comparison with the series expansion of $\zeta(3)$ (where $\zeta$
is the Riemann zeta function). This series converges to
\[
\langle S_{2,2}\rangle=\frac{1}{3},
\]
which agrees with Page's explicit formula \cite{Page1993}.

We can now apply this method to the terms in the general series.

\section{The general case}

The convergence of the series is trivial to establish when $m=1$,
so we will state that first:
\begin{lem}
\label{lem:m=00003D1}The series (\ref{eq:Lubkin-series}) is trivial,
and so converges absolutely, when $m=1$.\end{lem}
\begin{proof}
When $m=1$, the reduced density operator $\hat{\rho}_{1}^{1,n}$
is necessarily just the one-dimensional identity (it acts on a one-dimensional
Hilbert space and its trace is unity, and the identity is the only
operator that satisfies these conditions). This also means that $\hat{\rho}_{1}^{1,n}=\hat{\rho}_{0}$,
so
\[
\langle\text{Tr}[(\hat{\rho}_{s}-\hat{\rho}_{0})^{r}]\rangle=0
\]
for any $r\ge1$. The terms in (\ref{eq:Lubkin-series}) are therefore
trivially zero, giving
\[
\langle S_{1,n}\rangle=\ln1=0.
\]

\end{proof}
For cases where $m\ge2$, the following result relating to the convergence
of series will also be required:
\begin{lem}
\label{lem:sum-convergence}If a sequence $A_{r}$ is defined for
all integers $r\ge1$ with the form
\[
A_{r}=\sum_{i=1}^{N}A_{r,i}
\]
for a finite constant $N$, and there exists a second sequence $B_{r}$
such that
\begin{equation}
\lim_{r\rightarrow\infty}\left|\frac{A_{r,i}}{B_{r}}\right|\le c\delta_{ik}\label{eq:partial-limit}
\end{equation}
for finite $c$ and $1\le k\le N$, then
\begin{equation}
\sum_{r=1}^{\infty}A_{r}\label{eq:infinite-A}
\end{equation}
converges absolutely to a finite value if and only if 
\begin{equation}
\sum_{r=1}^{\infty}B_{r}\label{eq:infininite-B}
\end{equation}
also converges.\end{lem}
\begin{proof}
Given (\ref{eq:partial-limit}), it is easy to see that
\[
\lim_{r\rightarrow\infty}\left|\frac{A_{r}}{B_{r}}\right|=\sum_{i=1}^{N}c\delta_{ik}=c,
\]
which is finite. It then follows from the limit comparison test that
(\ref{eq:infinite-A}) converges absolutely if and only if (\ref{eq:infininite-B})
converges absolutely.
\end{proof}
We now have all the necessary tools to establish the conditions under
which the general series converges and diverges.
\begin{thm}
\label{thm:m>=00003D2}The series (\ref{eq:Lubkin-series}) converges
if and only if $m\le2$.\end{thm}
\begin{proof}
Only cases where $m\ge2$ need be considered here due to Lemma (\ref{lem:m=00003D1}),
so we will assume during this proof that $m\ge2$.

(\ref{eq:integer-powers-expr}) states that
\[
\langle\text{Tr}[(\hat{\rho}_{m}^{mn})^{r}]\rangle=\frac{\Gamma(mn)r!}{\Gamma(mn+r)}\sum_{k=0}^{m-1}\binom{m}{k+1}\binom{r-1}{k}\binom{n+r-k-1}{n-1}
\]
for any integer $r\ge1$. In addition it is known that
\[
\langle\text{Tr}[(\hat{\rho}_{m}^{mn})^{0}]\rangle=a_{0},
\]
where $a_{0}=\min(m,n)$ (see the footnote on page \pageref{fn:a_0}).
Therefore,
\begin{eqnarray}
\langle\text{Tr}[(\hat{\rho}_{m}^{mn}-\hat{\rho}_{0})^{r}]\rangle=\sum_{k=0}^{r}\binom{r}{k}(-1)^{r-k}\frac{1}{m^{r-k}}\langle\text{Tr}[\hat{\rho}_{m}^{mn}]\rangle\nonumber\\
\hphantom{\langle\text{Tr}[(\hat{\rho}_{m}^{mn}-\hat{\rho}_{0})^{r}]\rangle}=\frac{(-1)^{r}a_{0}}{m^{r}}+\sum_{k=1}^{r}\binom{r}{k}(-1)^{r-k}\frac{1}{m^{r-k}}\frac{\Gamma(mn)k!}{\Gamma(mn+k)}\nonumber\\
\times\sum_{s=0}^{m-1}\binom{m}{s+1}\binom{k-1}{s}\binom{n+k-s-1}{n-1}\nonumber\\
\hphantom{\langle\text{Tr}[(\hat{\rho}_{m}^{mn}-\hat{\rho}_{0})^{r}]\rangle}=\frac{(-1)^{r}a_{0}}{m^{r}}+\frac{\Gamma(mn)r!}{\Gamma(mn+r)}\sum_{k=1}^{r}\binom{r+mn-1}{k+mn-1}\frac{(-1)^{r-k}}{m^{r-k}}\nonumber\\
\times\sum_{s=0}^{m-1}\binom{m}{s+1}\binom{k-1}{s}\binom{n+k-s-1}{n-1}\nonumber\\
\hphantom{\langle\text{Tr}[(\hat{\rho}_{m}^{mn}-\hat{\rho}_{0})^{r}]\rangle}=\frac{(-1)^{r}a_{0}}{m^{r}}+\frac{\Gamma(mn)r!}{\Gamma(mn+r)}\sum_{k=mn}^{r+mn-1}\binom{r+mn-1}{k}\frac{(-1)^{r-k+mn-1}}{m^{r-k+mn-1}}\nonumber\\
\times\sum_{s=0}^{m-1}\binom{m}{s+1}\binom{k-mn}{s}\binom{n+k-mn-s}{n-1},\label{eq:offset-trace-expansion}
\end{eqnarray}
The various binomial coefficients at the end can be simplified using
the identity
\[
\frac{1}{(n-1)!}\frac{\partial^{n-1}}{\partial u^{n-1}}\left.\left[\frac{u^{n}}{s!}\frac{\partial^{s}}{\partial u^{s}}u^{k-mn}\right]\right|_{u=1}=\binom{k-mn}{s}\binom{n+k-mn-s}{n-1}.
\]
Substituting this into (\ref{eq:offset-trace-expansion}), we can
rearrange the summation over $k$ by the same procedure used in Section
\ref{sec:Special-Case}, giving 
\begin{eqnarray*}
\langle\text{Tr}[(\hat{\rho}_{m}^{mn}-\hat{\rho}_{0})^{r}]\rangle=\frac{(-1)^{r}a_{0}}{m^{r}}+\frac{\Gamma(mn)r!}{\Gamma(mn+r)}\sum_{s=0}^{m-1}\binom{m}{s+1}\frac{1}{(n-1)!}\\
\times\frac{\partial^{n-1}}{\partial u^{n-1}}\left[\frac{u^{n}}{s!}\frac{\partial^{s}}{\partial u^{s}}\left(\frac{1}{u^{mn}}\left(u-\frac{1}{m}\right)^{r+mn-1}\right.\right.\\
\left.\left.\left.-\sum_{k=0}^{mn-1}\binom{r+mn-1}{k}(-1)^{r-k+mn-1}\frac{u^{k-mn}}{m^{r-k+mn-1}}\right)\right]\right|_{u=1},
\end{eqnarray*}
and then expanding the two derivatives gives
\begin{eqnarray}
\langle\text{Tr}[(\hat{\rho}_{m}^{mn}-\hat{\rho}_{0})^{r}]\rangle=\frac{(-1)^{r}a_{0}}{m^{r}}+\sum_{s=0}^{m-1}\binom{m}{s+1}\frac{(-1)^{s}}{(n-1)!s!}\sum_{q=0}^{n-1}\binom{n-1}{q}\frac{n!(-1)^{q}}{(q+1)!}\nonumber\\
\times\left[\sum_{k=0}^{q+s}\binom{q+s}{k}(-1)^{k}\frac{r!(q+s+mn-k-1)!}{(r+mn-k-1)!}\left(1-\frac{1}{m}\right)^{r+mn-k-1}\right.\nonumber\\
\left.-\sum_{k=0}^{mn-1}\binom{mn-1}{k}\frac{r!(q+s+mn-k-1)!}{(r+mn-k-1)!}\frac{(-1)^{r-k+mn-1}}{m^{r-k+mn-1}}\right].
\end{eqnarray}
Finally, this gives the exact form for the terms in (\ref{eq:Lubkin-series})
(labelled $T_{r}$ for simplicity) as
\begin{eqnarray}
T_{r}=\frac{(-1)^{r}m^{r}}{r(r+1)}\langle\text{Tr}[(\hat{\rho}_{m}^{mn}-\hat{\rho}_{0})^{r+1}]\rangle\nonumber\\
\hphantom{T_{r}}=-\frac{a_{0}}{mr(r+1)}+\sum_{s=0}^{m-1}\binom{m-1}{s}\frac{m(-1)^{s}}{(s+1)!}\sum_{q=0}^{n-1}\binom{n-1}{q}\frac{n(-1)^{r+q}}{(q+1)!}\nonumber\\
\times\left[\sum_{k=0}^{q+s}\binom{q+s}{k}(-1)^{k}\frac{(r-1)!(q+s+mn-k-1)!}{(r+mn-k)!}\frac{(m-1)^{r+mn-k}}{m^{mn-k}}\right.\nonumber\\
\left.-\sum_{k=0}^{mn-1}\binom{mn-1}{k}\frac{(r-1)!(q+s+mn-k-1)!}{(r+mn-k)!}\frac{(-1)^{mn-k}}{m^{mn-k}}\right].\label{eq:Lubkin-term}
\end{eqnarray}
This expression is now compatible with Lemma \ref{lem:sum-convergence},
as it gives each term in Lubkin's series as a sum over a fixed number
of terms. Therefore, to determine if (\ref{eq:Lubkin-series}) converges
we only need to determine what the limiting behaviour of the dominant
term in (\ref{eq:Lubkin-term}) is.

When $m=2$, the limiting behaviour is given by the sequence $B_{r}^{(2)}=r^{-2}$,
which is comparable only to $(r-1)!/(r+2n-k)!$ with $k=2n-1$, as
well as to the $a_{0}$ term %
\footnote{Lemma \ref{lem:sum-convergence} specifically requires it to be comparable
to a single term, but the two terms here can be summed to produce
a single term which is still comparable with $B_{r}^{(2)}$.%
}. When $m>2$, the limiting behaviour is given by 
\[
B_{r}^{(m)}=\frac{(m-1)^{r}}{r^{(m-1)(n-1)+2}},
\]
which is comparable to $(r-1)!(m-1)^{r}/(r+mn-k)!$ when $k=m+n-2$.
It therefore follows from the fact that 
\[
\sum_{r=1}^{\infty}B_{r}^{(2)}=\frac{1}{\zeta(2)}
\]
converges absolutely but
\[
\sum_{r=1}^{\infty}B_{r}^{(m)}
\]
does not (in addition to Lemma \ref{lem:m=00003D1}), that (\ref{eq:Lubkin-series})
converges if and only if $m\le2$.
\end{proof}

\section{Discussion and Conclusions}

Lubkin's original derivation of his approximation was based on the
assumption that (\ref{eq:Lubkin-series}) converged quickly to a finite
value, allowing truncations of the series to be used as approximations.
However, we have now proved that the only cases of significance where
the series converges %
\footnote{The series also converges when $m=1$, but the convergence is trivial
and the entropy is zero, so this is not of any practical benefit.%
} are those where $m=2$, Lubkin's series might still be used to compute
the entropy in these cases then, but it will only be of use in single-qbit
systems. In all other cases, Lubkin's series diverges rapidly, so
being equal to a truncation of this series is not sufficient proof
that Lubkin's guess is a good approximation for the entropy. This
is surprising, however, as we are able to confirm the validity of
Lubkin's guess by other methods (see Appendix 3).

It is worth noting, however, that the proof of Lubkin's approximation
given in Appendix 3 also gives us and alternative (and better) approximation
to the entropy i.e.
\[
\langle S_{m,n}\rangle\simeq\ln2-\frac{m^{2}-1}{2mn}
\]
when $n\gg1$, which has a broader range of validity \emph{and} a
smaller error in general than Lubkin's guess.

The work shown in this paper, specifically the closed-form expressions
for $\langle\text{Tr}[(\hat{\rho}_{m}^{mn})^{r}]\rangle$ found in
Section \ref{sec:Preliminaries}, have an additional application;
as we will prove in another paper \cite{Dyer2014a}, $\langle\text{Tr}[(\hat{\rho}_{m}^{mn})^{r}]\rangle$
can be used to find a generating function for enumerating sets of
combinatorial hypermaps. Specifically,
\begin{eqnarray*}
\langle\text{Tr}[(\hat{\rho}_{m}^{mn})^{r}]\rangle & = & \frac{\Gamma(mn)}{\Gamma(mn+r)}\sum_{e,v}h_{r}^{(1)}(e,v)m^{e}n^{v}\\
 & \equiv & \frac{\Gamma(mn)}{\Gamma(mn+r)}P_{r}(m,n)
\end{eqnarray*}
for any positive integers $m$, $n$ and $r$, where $h_{r}^{(1)}(e,v)$
is the number of rooted hypermaps with one face, $e$ edges and $v$
vertices (these objects are defined in that paper, and also discussed
in detail in \cite{Lando2004}). Expressions such as (\ref{eq:real-powers-expr})
then allow us to find closed-form expressions for the generating function
$P_{r}(m,n)$ which generates $h_{r}^{(1)}(e,v)$. Generating functions
of this type are a very powerful tool in enumerative combinatorics,
as they have strong connections to the underlying structure of the
classes they enumerate \cite{Stanley1997}, so being able to express
them in closed form in this way is a significant result.

\section*{Acknowledgements}

This problem was suggested to me by Bernard Kay, who had begun the
process of evaluating terms in Lubkin's series by a different method,
and I am grateful to him for sharing his unpublished work on that
with me, as well as providing advice and help while writing the paper.

This work was supported by an EPSRC-funded studentship through the
Department of Mathematics at the University of York.

\section*{Appendix 1: An alternative proof of Page's exact entropy formula}

In this appendix we re-derive Page's explicit von Neumann entropy
formula \cite{Page1993}. This method demonstrates the parallel between
the proof of Theorem \ref{thm:real-powers} and the method previously
used by Sen for this purpose \cite{Sen1996}.

As stated in (\ref{eq:real-powers-expr}),
\[
\langle\text{Tr}[(\hat{\rho}_{m}^{mn})^{r}]\rangle=\frac{\Gamma(mn)}{r\Gamma(mn+r)}\sum_{k=0}^{m-1}\frac{(-1)^{k}\Gamma(m+r-k)\Gamma(n+r-k)}{k!\Gamma(r-k)\Gamma(m-k)\Gamma(n-k)}
\]
for general $r$. We can use this to evaluate the mean von Neumann
entropy using the fact that
\begin{eqnarray*}
\langle S_{m,n}\rangle & = & -\langle\text{Tr}[\hat{\rho}_{m}^{mn}\ln\hat{\rho}_{m}^{mn}]\rangle\\
 & = & -\lim_{r\rightarrow1}\frac{\partial}{\partial r}\langle\text{Tr}[(\hat{\rho}_{m}^{mn})^{r}]\rangle.
\end{eqnarray*}
First, differentiating gives
\begin{eqnarray*}
-\frac{\partial}{\partial r}\langle\text{Tr}[(\hat{\rho}_{m}^{mn})^{r}]\rangle & = & \frac{\Gamma(mn)}{r\Gamma(mn+r)}\sum_{k=0}^{m-1}\frac{(-1)^{k}\Gamma(m+r-k)\Gamma(n+r-k)}{k!\Gamma(r-k)\Gamma(m-k)\Gamma(n-k)}\\
 &  & \times\left(\psi(mn+r)-\psi(m+r-k)-\psi(n+r-k)+\psi(r-k)+\frac{1}{r}\right).
\end{eqnarray*}
When we take the limit of $r\rightarrow1$, the $1/\Gamma(r-k)$ factor
causes most terms to vanish in cases when $k>0$. The only ones which
remain are those where $\psi(r-k)$ is in the numerator:
\begin{eqnarray*}
\langle S_{m,n}\rangle & = & -\lim_{r\rightarrow1}\frac{\partial}{\partial r}\langle\text{Tr}[\hat{\rho}_{s}^{r}]\rangle\\
 & = & \psi(mn+1)-\psi(m+1)-\psi(n+1)+\psi(1)+1\\
 &  & +\frac{1}{mn}\sum_{k=1}^{m-1}(m-k)(n-k)\lim_{r\rightarrow1}\frac{(-1)^{k}\psi(r-k)}{\Gamma(k+1)\Gamma(r-k)}\\
 & = & \psi(mn+1)-\psi(m+1)-\psi(n+1)+\psi(1)+1\\
 &  & +\frac{1}{mn}\sum_{k=1}^{m-1}\frac{(m-k)(n-k)}{k},
\end{eqnarray*}
where here we use the fact that 
\begin{eqnarray*}
\lim_{r\rightarrow1}\frac{\psi(r-k)}{\Gamma(r-k)} & = & \lim_{r\rightarrow1}\frac{1}{[\Gamma(r-k)]^{2}}\frac{\partial}{\partial r}\Gamma(r-k)\\
 & = & -\lim_{r\rightarrow1}\frac{1}{\text{Res}_{x=1-k}[\Gamma(x)]}\\
 & = & (-1)^{k}\Gamma(k)
\end{eqnarray*}
for positive integers $k$.

Then we use the identity
\[
\psi(b)-\psi(a)=\sum_{k=a}^{b-1}\frac{1}{k},
\]
on the pairs $(\psi(mn+1)-\psi(n+1))$ and $(\psi(m-1)-\psi(1))$,
giving
\[
\langle S_{m,n}\rangle=\sum_{k=n+1}^{mn}\frac{1}{k}-\frac{m-1}{2n},
\]
which agrees with Page's formula \cite{Page1993}.

\section*{Appendix 2: A gamma-function summation identity}
\begin{lem}
\label{lem:Gamma-lemma}
\begin{eqnarray}
\sum_{i=0}^{a}\frac{(-1)^{i}}{\Gamma(i+1)\Gamma(r-i+1)} & = & \frac{(-1)^{a}}{r\Gamma(a+1)\Gamma(r-a)}\label{eq:lemma-ident}
\end{eqnarray}
for any integer $a\ge0$.\end{lem}
\begin{proof}
(\ref{eq:lemma-ident}) holds when $a=0$, as both sides of the equation
simply equal $1/\Gamma(r+1)$ in that case. Now if we assume that
(\ref{eq:lemma-ident}) holds for $a=N\ge0$, extending the sum to
$a=N+1$ gives
\begin{eqnarray*}
\sum_{i=0}^{N+1}\frac{(-1)^{i}}{\Gamma(i+1)\Gamma(r-i+1)} & = & \frac{(-1)^{N}}{r\Gamma(N+1)\Gamma(r-N)}+\frac{(-1)^{N+1}}{\Gamma(N+2)\Gamma(r-N)}\\
 & = & \left(1-\frac{N+1}{r}\right)\frac{(-1)^{N+1}}{\Gamma(N+2)\Gamma(r-N)}\\
 & = & \frac{(-1)^{N+1}}{r\Gamma(N+2)\Gamma(r-N-1)},
\end{eqnarray*}
which also agrees with (\ref{eq:lemma-ident}). It therefore follows
by induction that (\ref{eq:lemma-ident}) holds for any $a\ge0$.
\end{proof}

\section*{Appendix 3: Confirmation of Lubkin's approximation using Page's formula}

Lubkin's proposed approximation for the entropy is \cite{Lubkin1978,Page1993}
\[
\langle S_{m,n}\rangle\simeq\ln m-\frac{m^{2}-1}{2mn+2}
\]
when $m\gg n$. However, the exact meaning of the `goodness' of this
approximation was left ambiguous. Here we prove that
\[
\langle S_{m,n}\rangle\simeq\ln m-\frac{m^{2}-1}{2mn}+\mathcal{O}\left(\frac{1}{n^{2}}\right),
\]
 and use this to prove the that Lubkin's guess is a good approximation
(although still not ideal).

The exact formula for $S_{m,n}$ as conjectured by Page \cite{Page1993}
and later proven by various authors \cite{Foong1994,Sanchez-Ruiz1995,Sen1996}
is
\[
\langle S_{m,n}\rangle=\sum_{k=n+1}^{mn}\frac{1}{k}-\frac{m-1}{2n}
\]
for $n\ge m$. This can equivalently be written
\begin{equation}
\langle S_{m,n}\rangle=H_{mn}-H_{n}-\frac{m-1}{2n},\label{eq:Harmonic-entropy}
\end{equation}
where $H_{n}$ is the $n^{\mbox{th}}$ harmonic number. $H_{n}$ has
an asymptotic expansion \cite{Conway1996}
\[
H_{n}=\ln n+\gamma+\frac{1}{2n}+\mathcal{O}\left(\frac{1}{n^{2}}\right),
\]
where $\gamma$ is Euler's constant. Substituting this into (\ref{eq:Harmonic-entropy}),
we get
\begin{eqnarray}
\langle S_{m,n}\rangle & = & \ln mn-\ln n+\frac{1}{2mn}-\frac{1}{2n}-\frac{m-1}{2n}+\mathcal{O}\left(\frac{1}{m^{2}n^{2}}\right)+\mathcal{O}\left(\frac{1}{n^{2}}\right)\nonumber \\
 & = & \ln m-\frac{m^{2}-1}{2mn}+\mathcal{O}\left(\frac{1}{n^{2}}\right).\label{eq:S-bigO}
\end{eqnarray}
This method is the origin of the asymptotic expansion Page mentions
in his paper \cite{Page1993}.

Explicitly, (\ref{eq:S-bigO}) means that there is some positive constant
$K$ such that
\[
\left|\langle S_{m,n}\rangle-\ln m+\frac{m^{2}-1}{2mn}\right|\le\frac{K}{n^{2}}
\]
 for all $m$ and $n$. To get an idea of the scale of $K$, we evaluate
the remainder using additional terms from the asymptotic expansion
of $H_{n}$. Doing this, we get
\[
\delta=\langle S_{m,n}\rangle-\ln m+\frac{m^{2}-1}{2mn}=\frac{m^{2}-1}{12m^{2}n^{2}}+\mathcal{O}\left(\frac{1}{n^{3}}\right).
\]
This suggests that $K$ is $1/12$, and numerical calculation of the
exact error in all cases with $m\le n$ up to $n=100$ shows the absolute
error increasing monotonically in both $m$ and $n$ up to $0.083324$
at $m=n=100$, apparently tending towards $1/12$. We can therefore
conclude that 
\begin{equation}
\langle S_{m,n}\rangle\simeq\ln-\frac{m^{2}-1}{2mn}\label{eq:my-approx}
\end{equation}
is a good approximation when $n\gg1$ with an error of $\mathcal{O}(1/n^{2})$
(the $n\gg1$ condition is to ensure that the error is small).

Lubkin's approximation differs from (\ref{eq:my-approx}), but when
the denominator in the $(m^{2}-1)/(2mn+2)$ term is expanded binomially,
the difference between the two approximations is iteslf $\mathcal{O}(1/n^{2})$.
The error in Lubkin's approximation is therefore
\[
\delta_{Lubkin}=\langle S_{m,n}\rangle-\ln m+\frac{m^{2}-1}{2mn+2}=\delta-\frac{m^{2}-1}{2m^{2}n^{2}}=-\frac{5(m^{2}-1)}{12m^{2}n^{2}}+\mathcal{O}\left(\frac{1}{n^{3}}\right).
\]
We therefore conclude that Lubkin's guess is \emph{also} a good approximation,
although the error is approximately five times larger than the error
in (\ref{eq:my-approx}).

While Lubkin's original guess was given the condition $m\ll n$, it
is now clear that all that is needed is that $n\gg1$ (the statement
(\ref{eq:S-bigO}) is true in general%
\footnote{On a related note, expanding $H_{mn}$ instead as
\[
H_{mn}=\ln mn+\mathcal{O}\left(\frac{1}{mn}\right),
\]
and seeing that $\mathcal{O}(1/n^{2})$ implies $\mathcal{O}(1/(mn))$
when $m\le n$, we get that
\begin{eqnarray}
\langle S_{m,n}\rangle & = & \ln mn-\ln n-\frac{1}{2n}-\frac{m-1}{2n}+\mathcal{O}\left(\frac{1}{mn}\right)+\mathcal{O}\left(\frac{1}{n^{2}}\right)\nonumber \\
 & = & \ln m-\frac{m}{2n}+\mathcal{O}\left(\frac{1}{mn}\right).\label{eq:O(1/mn)}
\end{eqnarray}
Page states a result similar to this in his paper, saying that the
average deviation of the entropy from the maximal entropy $S_{m}^{max}=\ln m$
is
\[
I_{m,n}=\ln m-\langle S_{m,n}\rangle=\frac{m}{2n}+\mathcal{O}\left(\frac{1}{mn}\right)
\]
when both $m$ and $n$ are large (i.e. $1\ll m\le n$) \cite{Page1993}.
(\ref{eq:O(1/mn)}) shows that this is in fact true in general, and
not just when the parameters are large.%
}, but the $n^{-2}$ dependence of the remainder means the error in
(\ref{eq:my-approx}) is smallest -- and therefore the approximation
is most accurate -- when $n$ is large), and requiring that $m$ be
small in comparison to $n$ has little effect on the accuracy. Thus,
while we can now see that Lubkin's approximation is itself good, it
is improved upon by the approximation (\ref{eq:my-approx}).

\bibliographystyle{spphys}
\bibliography{Lubkin_Convergence}

\end{document}